\documentclass[]{amsart}

\usepackage{tikz}
\usepackage{subfigure}

\theoremstyle{plain}
\newtheorem{thm}{Theorem}[section]
\newtheorem{theorem}[thm]{Theorem}

\newtheorem{lemma}[thm]{Lemma}

\theoremstyle{definition}
\newtheorem{remark}[thm]{Remark}

\newtheorem{definition}[thm]{Definition}

\numberwithin{equation}{section}

\newcommand{\Log}{\mbox{{\sf L}}}
\newcommand{\ULog}{\mbox{{\sf UL}}}
\newcommand{\NLog}{\mbox{{\sf NL}}}
\newcommand{\Ptime}{\mbox{{\sf P}}}

\begin{document}

\title[Tractable Win-Lose Games]{Some Tractable Win-Lose Games}

\keywords{Win-lose bimatrix game, $K_{3,3}$-minor-free, $K_{5}$-minor-free, Nash equilibrium}

\author[S. Datta]{}
\author[N. Krishnamurthy]{}

\maketitle

\begin{center}
$\begin{array}{cc}{\mbox{\sc {Samir Datta}}} & {\mbox{\sc {Nagarajan Krishnamurthy}}} \\ {sdatta@cmi.ac.in}  & {naga@cmi.ac.in} \\ & \\ \end{array}$
\end{center}

\begin{center}
$\begin{array}{c}
\mbox{Chennai Mathematical Institute}\\
\mbox{India.}\\ 
\\  
\end{array}$
\end{center}


\begin{abstract}
 \noindent Determining a Nash equilibrium in a $2$-player non-zero sum game is known to be PPAD-hard (Chen and Deng, 2006 \cite{CD06}, Chen, Deng and Teng 2009 \cite{CDT09}). The problem, even when restricted to win-lose bimatrix games, remains PPAD-hard  (Abbott, Kane and Valiant, 2005 \cite{AKV05}). However, there do exist polynomial time tractable classes of win-lose bimatrix games - such as, very sparse games (Codenotti, Leoncini and Resta, 2006 \cite{CLR06}) and planar games (Addario-Berry, Olver and Vetta, 2007 \cite{AOV07}). 

We extend the results in the latter work to $K_{3,3}$ minor-free games and a subclass of $K_5$ minor-free games. Both these classes of games strictly contain planar games. Further, we sharpen the upper bound to unambiguous logspace, a small complexity class contained well within polynomial time. Apart from these classes of games, our results also extend to a class of games that contain both $K_{3,3}$ and $K_5$ as minors, thereby covering a large and non-trivial class of win-lose bimatrix games. For this class, we prove an upper bound of nondeterministic logspace, again a small complexity class within polynomial time. Our techniques are primarily graph theoretic and use structural characterizations of the considered minor-closed families. 

\end{abstract}


\section{Introduction}\label{section:intro}

Nash, in a seminal paper in 1951 \cite{Na51}, proved that every finite non-cooperative game always has at least one equilibrium point. In recent years, there has been a flurry of activity in the algorithmic and complexity theoretic community addressing various questions regarding the computation of Nash equilibria. On the one hand, related hardness results have been proved. Daskalakis, Goldberg and Papadimitriou \cite{DGP09} showed that finding a Nash equilibrium is PPAD complete for an $n$-player game ($n \geq 4$). (Refer Papadimitriou \cite{PPAD94} for a definition of the complexity class PPAD). Chen and Deng \cite {CD05} extended this result to the $3$-player case. The belief that the $2$-player case is polynomial time tractable continued, till Chen and Deng \cite {CD06} (also refer Chen, Deng and Teng \cite{CDT09}) proved PPAD completeness for the $2$-player case too. Abbott, Kane and Valiant \cite{AKV05} showed that restricting the payoffs of the players to 0 or 1 does not make the game easier. In other words, bimatrix $(0,1)$ or win-lose games are as hard as the general-sum case. Chen, Teng and Valiant \cite{CTV07} showed that even finding approximate Nash equilibria (correct upto a logarithmic number of bits) for win-lose bimatrix games is PPAD hard.  

On the other hand, search is on for classes of games where efficient (polynomial time, say) algorithms exist. Codenotti, Leoncini and Resta \cite{CLR06} proposed an algorithm to efficiently compute Nash equilibria in sparse win-lose bimatrix games. They showed that win-lose bimatrix games that have at most two winning positions per pure strategy can be solved in linear time. Addario-Berry, Olver and Vetta \cite{AOV07} showed that planar win-lose bimatrix games are polynomial time tractable. Planar win-lose bimatrix games are games where the bipartite graph obtained by considering the payoff matrices as adjacency matrices is planar. See Section~\ref{section:background} for formal definitions. They proved that:

\begin{theorem}(Addario-Berry,Olver and Vetta, 2007 \cite{AOV07}, Theorem 3.6)
There is a polynomial time algorithm for finding a Nash equilibrium in a two-player planar win-lose game.
\end{theorem}

We sharpen the above result by observing that polynomial time can be replaced by deterministic logarithmic space \Log\ and more importantly, we extend the result from planar games to the following classes of win-lose bimatrix games where we prove solvability in unambiguous logarithmic space \ULog\ for classes (1) and (2), and in nondeterministic logarithmic space \NLog\ for class (3). (Recall that \ULog\ $\subseteq$ \NLog). 
\begin{enumerate}
\item $K_{3,3}$ minor-free games, 
\item $K_5$-minor free games where the triconnected components are either planar or $V_8$, 
\item Games whose triconnected components are $K_5$, $V_8$ or planar. 
\end{enumerate}

For definitions of these classes, see Section~\ref{section:background}. Notice that all these classes strictly include planar graphs since from Kuratowski's Theorem \cite{Kur30} (also refer \cite{Wag37} by Wagner), it follows that planar graphs are exactly the class of graphs which exclude both $K_{3,3}$ and $K_5$ as minors. The class of games in (3) above strictly contains those in (1) and (2) (Ref. \cite{As85}, \cite{Wag37}), and also contains games that are neither $K_{3,3}$-minor-free nor $K_5$-minor-free. Hence, our results cover a large and non-trivial superclass of planar win-lose bimatrix games. Our results are motivated by a remark in \cite{AOV07} which indicates that planarity is a necessary condition for their method to work and exhibits an oriented $K_{3,3}$ as an example where their method does not apply. 

Our proofs for all these classes use the essential ingredients of the proof in \cite{AOV07}, in particular, showing that an undominated induced cycle exists for each of these classes. Such a cycle corresponds to a Nash equilibrium. (We outline these results from \cite{AOV07}, in Section~\ref{section:background}). We also show that such a cycle can be found in \ULog\ for classes (1) and (2), and in \NLog\ for class (3) (and hence, in polynomial time for all three classes). Our proofs, further, build on a triconnected decomposition of planar graphs and $K_{3,3}$-minor-free graphs (respectively, $K_5$-minor-free graphs). 

We discuss the proof outline in the following section. For details of these proofs, refer to Section~\ref{section:mainproof}. Section~\ref{section:background} contains background and preliminaries, and Section~\ref{section:futurework} concludes with some open problems. 

\section{Outline of Proof}\label{section:proofoutline}

\subsection{$K_{3,3}$ minor-free games}

Given a $K_{3,3}$ minor-free game we consider its underlying graph. The essential idea carried over from \cite{AOV07} is to identify an \emph{undominated induced cycle} in this graph. Such a cycle corresponds to a Nash equilibrium. (Refer Section~\ref{section:background} for definitions). 

Since we are looking for cycles of a particular kind we need to focus on any biconnected component of the underlying undirected graph (since even undirected cycles cannot span biconnected components). We further decompose the underlying biconnected graph into triconnected components using a method from \cite{DNTW09}. We show that if we can produce undominated induced cycles in each triconnected component (which inherits orientation to edges from the original graph), then there is a way to ``stitch'' them together to obtain at least one undominated induced cycle in the original graph.

The problem thus reduces to finding undominated induced cycles in the triconnected components of a $K_{3,3}$ minor-free graph which by a lemma of Asano \cite{As85} are exactly triconnected planar graphs or $K_5$'s. For the former, we know from \cite{AOV07} how to find such cycles and for the latter, we explicitly show how to find such cycles.

Notice that in the process of finding triconnected components we have lost bipartiteness and possibly altered the notion of domination. So we have to carefully deal with subdivisons of triconnected graphs, instead which preserve bipartiteness, domination and most properties of $3$-connectivity.

\subsection{Subclass of $K_{5}$ minor-free games}

Similar to the $K_{3,3}$-case outlined above, we consider the graph corresponding to the given $K_5$-minor-free game. In addition to being $K_5$-minor free, the triconnected components of the underlying graph are either planar or $V_8$. For this subclass, we prove the existence of an \emph{undominated induced cycle} and show how to find one. 

Again, it suffices to consider any biconnected component of the underlying undirected graph and we use a method from \cite{DNTW09} to further decompose the biconnected component into triconnected components. These triconnected components are either planar, $V_8$ (the four rung Mobius ladder) or its $4$-connected components are all planar (Ref. Wagner \cite{Wag37}, Khuller \cite{Khu88}). We prove that every strongly connected orientation of $V_8$ has an undominated induced cycle as well. As the triconnected components are either planar or $V_8$, each component has an undominated induced cycle and we show that there is a way to ``stitch'' them together to obtain an undominated induced cycle in the original graph. \\

The proof for games whose triconnected components are $K_5$, $V_8$ or planar follows from the proofs for the other two classes discussed above.   

\section{Background and Preliminaries}\label{section:background}

\subsection{Win-Lose Bimatrix Games}

A win-lose bimatrix game is a two-player game in normal form where the payoffs to the players are either $0$ or $1$. Given sets of actions of the players (the row player and the column player), the aim of each player is to choose a strategy that maximizes (or minimizes) his/ her expected payoff. We formally define these below. We shall assume that both the players want to maximize their respective payoffs. 

\begin{definition}
A win-lose (or $(0,1)$-) bimatrix game is specified by two $R \times C$ matrices $M_R$ and $M_C$ with entries from $\{0,1\}$ where the payoff for the row (column) player playing pure strategy $r_i \in R$ (respectively, $c_j \in C$) is $M_R(i,j)$ (respectively, $M_C(i,j)$) if the column (respectively, row) player plays strategy $c_j \in C$ (respectively, $r_i \in R$) in response.

A pure strategy of a player is an action choice of the player. Pure strategies of the row (column) player are just rows (respectively, columns) that index the payoff matrices of the players. A mixed strategy of the row (column) player is a probability distribution over the set of rows (respectively, columns). 
\end{definition}

\begin{definition}
A best response of the row (column) player is a mixed strategy $x$ (respectively, $y$) of the row (respectively, column) player that maximizes his (her) expected payoff given a mixed strategy $y$ (respectively $x$) of the other player. 

A Nash equilibrium is a pair of strategies that are mutual best responses. In other words, (mixed) strategies $x^*$ of player $1$ (the row player) and $y^*$ of player $2$ (the column player) constitute a Nash equilibrium if and only if

$x^{*^T} M_R \mbox{ }y^* \geq x^T M_R \mbox{ }y^*$, for all strategies $x$ of player $1$, and 

$x^{*^T} M_C \mbox{ }y^* \geq x^{*^T} M_C \mbox{ }y$, for all strategies $y$ of player $2$. 
\end{definition}

\begin{definition}
Given a win-lose bimatrix game specified with the matrices $M_R, M_C$, its associated (or underlying) graph is the bipartite directed graph, $G_{R,C} \mbox{ } = (V, E)$ with bipartitions $R, C$ (that is, $V = R \cup C$) and the following adjacency matrix:
\[
\left(
\begin{array}{cc}
0 &    M_R \\
M_C^T & 0 \\
\end{array}
\right)
\]
\end{definition}

\begin{definition}
An {\em undominated induced cycle} in $G_{R,C}$ is a cycle $\mathcal{C}$ such that there are no vertices $\{u, v, w\}$, where $v$ and $w$ ($v \not= w$) are on $\mathcal{C}$, such that $(u, v) \in \mathcal{C}$ and $(u, w) \in \mathcal{C}$. 

An induced cycle $\mathcal{C}$ is {\em dominated} by a vertex $v$ (not on $\mathcal{C}$), if there are two or more edges {\em from} $v$ {\em to} vertices on $\mathcal{C}$). 
\end{definition}

Addario-Berry, Olver and Vetta \cite{AOV07} showed that if the underlying bipartite directed graph ($G_{R,C}$ as defined above) of a win-lose bimatrix game is planar, a Nash equilibrium can be computed in polynomial time. We use the following results from \cite{AOV07}. \\

{\noindent}{\textbf{Claim 1. }}(Sections 2.1 and 2.2 in \cite{AOV07})
It suffices to look at the case when $G_{R,C}$ is strongly connected and is free of digons (cycles of length 2). 

\begin{lemma}\label{lemma:UICNash}(Corollary 2.3 in \cite{AOV07})
Let $S \subset V$ and let the subgraph restricted to $S$, $G_{R,C}[S]$, be an induced cycle. $S$ corresponds to a (uniform) Nash equilibrium if the cycle $G_{R,C}[S]$ is undominated. 
\end{lemma}

\begin{lemma}\label{lemma:UCinPlan}(Theorem 3.4 in \cite{AOV07})
Any non-trivial, strongly connected, bipartite, planar graph contains an undominated facial cycle.
\end{lemma}

\begin{lemma}\label{lemma:mainInAOV}(Theorem 3.1 in \cite{AOV07})
Any non-trivial, strongly connected, bipartite, planar graph has an undominated induced cycle. 
\end{lemma} 

\subsection{Minor-free Graphs}

\begin{definition}
A subdivision of a graph $G$ is a graph obtained by subdividing its edges by adding more vertices. 
\end{definition}

\begin{definition} 
Given a graph $G$, a minor of $G$ is a subgraph $H$ that can be obtained from $G$ by a finite sequence of edge-removal and edge-contraction operations. If G does not contain $K_{3,3}$ ($K_5$) as a minor, it is called $K_{3,3}$-minor-free ($K_5$-minor-free). 
\end{definition}

We are interested in $K_{3,3}$-minor-free graphs and $K_5$-minor-free graphs due to Kuratowski's (and Wagner's) theorems. 

\begin{theorem} (Kuratowski, 1930 \cite{Kur30}, Wagner, 1937 \cite{Wag37})
A finite graph is planar if and only if it contains neither a $K_{3,3}$-minor nor a $K_5$-minor. 
\end{theorem}

\begin{definition} 
A $k$-connected graph is one which remains connected on removing less than $k$ vertices. A $2$-connected graph is called biconnected and $3$-connected graphs together with cycles constitute triconnected graphs. 
\end{definition} 

\begin{definition}
Consider two graphs $G_1$ and $G_2$ each containing cliques of equal size. The clique-sum of $G_1$ and $G_2$ is a graph formed from their disjoint union by identifying pairs of vertices in the corresponding equal-sized cliques to form a shared clique, and possibly deleting some of the clique edges. A $k$-clique-sum is a clique-sum in which both cliques have at most $k$-vertices. 
\end{definition}

We use the following results regarding $K_{3,3}$-minor-free and $K_5$-minor-free graphs. 

\begin{lemma}\label{lemma:K33_tricon} (Asano, 1985 \cite{As85}. Also refer Vazirani, 1989 \cite{Vaz89})
Let $G$ be a $K_{3,3}$-minor-free graph. The triconnected components of $G$ are either planar or $K_5$. 
\end{lemma} 

\begin{lemma}\label{lemma:K5_tricon} (Wagner, 1937 \cite{Wag37}. Also refer Khuller, 1988 \cite{Khu88})
Let $G$ be a $K_{5}$-minor-free graph. The non-planar triconnected components of $G$ are either $V_8$ or 3-clique-sums of $4$-connected planar graphs. In other words, all $K_{5}$-minor-free graphs can be obtained by repeatedly taking $3$-clique-sums of planar graphs and $V_8$. 
\end{lemma} 

The class of $K_5$-minor free graphs where we prove existence of undominated induced cycles, namely $K_5$-minor free graphs whose triconnected components are either $V_8$ or planar, can also be obtained by taking $3$-clique-sums of planar graphs and $V_8$, but we require every clique-sum to be between two graphs. Repeatedly taking clique-sums may lead to graphs with no undominated induced cycles. For example, $K_{3,3}$ can be obtained by a 3-clique-sum operation on three $K_4$'s. We can orient these edges to obtain the counter example shown in \cite{AOV07}.  

\subsection{Triconnected decomposition of a graph}
Given a graph we can easily find its cut-vertices and biconnected (i.e. $2$-vertex connected) components, see e.g. \cite{DLNTW09} for how to do this in \Log. If we start with a digraph instead, there is virtually no change in the procedure since two biconnected components do not share any edge.

In 1973, Hopcroft and Tarjan \cite{HT73} proposed a technique to decompose a graph into its triconnected components\footnote{i.e. $3$-vertex connected components or cycles}. Though the procedure to find the triconnected components is not as simple as that for finding biconnected components, it can also be accomplished in \Log\ - see e.g. \cite{DNTW09}. As an output of this procedure we obtain a triconnected-component-separating-pair tree which has triconnected components and separating pairs\footnote{i.e. pairs of vertices removing which leaves the graph disconnected. Not all the separating pairs in a graph are present in the triconnected-component-separating-pair tree - only triconnected separating pairs - for details see \cite{DNTW09}} as nodes with an edge between a triconnected component and a separating pair if and only if the separating pair belongs to the component. 


\section{Finding a Nash Equilibrium}\label{section:mainproof}

In this section, we prove that we can find a Nash equilibrium in \ULog\ (and hence in polynomial time) for (1) $K_{3,3}$-minor-free win-lose bimatrix games, (2) $K_5$-minor-free win-lose bimatrix games with triconnected components being planar or $V_8$, and in \NLog\ (and hence in polynomial time) for (3) Win-lose bimatrix games whose triconnected components are planar, $K_5$ or $V_8$. 

From Claim 1, it suffices to look at bipartite digraphs that are strongly connected and free of digons. It is easy to see that this initial pre-processing can be done in \Log. If the graph $G_{R,C}$ has a digon, a pure equilibrium can easily be found in \Log. Checking if $G_{R,C}$ is strongly connected can be done in \Log\ as well. If $G_{R,C}$ is not strongly-connected, we can find one strongly connected component, say $S$, of $G_{R,C}$ such that there are no edges from any other component into $S$, in \Log. Let us call $S$ undominated, in this case. (If $H$ is the strongly-connected-component-dag, (strongly connected components of $G_{R,C}$ are vertices of $H$. For $S_1$, $S_2 \in H$, $(S_1, S_2)$ is an edge if $\exists v_1 \in S_1, v_2 \in S_2$ such that $(v_1, v_2)$ is an edge in $G_{R,C}$), then $S$ is a source in $H$). 

Now, we find the triconnected components of $G_{R,C}$ and as discussed in Section~\ref{section:background}, this can be done in \Log. It is important to notice that while computing a triconnected component, we add a spurious edge between the vertices of a separating pair if it is not already present. These spurious edges pose three kinds of problems when we consider bipartite digraphs:
\begin{itemize}
\item How do we orient these edges?
\item How do we ensure bipartiteness in the resulting triconnected components? 
\item How do we make sure that no new domination is introduced? 
\end{itemize}

For handling the first problem, notice that, if there is a directed path from $u$ to $v$, (where $\{u,v\}$ is a separating pair) in some component $C_0$ formed after removing $\{u,v\}$ from the graph, then we can orient the edge as $(u,v)$ in every other component $C$ formed by removing the same separating pair. This will ensure that if we find a cycle in a triconnected component it always corresponds to a cycle in the original graph. However, in order to do this orientation, we do need to check for reachability in the corresponding directed graphs, which is in \NLog. For $K_{3,3}$-minor-free and $K_5$-minor-free graphs, we know from \cite{TW09} that this is in \ULog. 

The second problem is easy to solve: we just need to subdivide the spurious edge appropriately i.e. if both its endpoints are in the same partition then subdivide the edge by introducing a single vertex, else subdivide it by introducing two vertices. 

Observe that, in order to solve the second problem, if both endpoints are in different partitions, we do not need to subdivide the edges at all. But this may introduce new dominations. Subdividing spurious edges even if the endpoints are in different partitions takes care of the third problem. (Note that as we are neither removing any edges nor subdividing existing edges, existing domination relationships are preserved). 

We summarize the result of the above discussion in the following lemma:

\begin{lemma}\label{lemma:decomp3Conn}
Given a digraph $G$ whose underlying undirected graph is biconnected, we can obtain a triconnected-component-separating-pair tree $T$ with edges in triconnected components directed in such a way that every directed cycle present in a triconnected component corresponds to some directed cycle in $G$. This procedure is in \NLog. If $G$ is $K_{3,3}$-minor-free or $K_5$-minor-free, then this procedure is in \ULog. \\
\end{lemma}

Now, we prove the following lemma. 

\begin{lemma}\label{lemma:UICin3connPlan}
Let $G$ be a strongly connected orientation of a subdivision of a triconnected planar graph which is bipartite. Then at least one of the faces of $G$ is bounded by an undominated, induced (directed) cycle. 
\end{lemma}

\begin{proof}
Using Lemma~\ref{lemma:UCinPlan} we know there exists an undominated, facial (directed) cycle. Since in a triconnected graph every facial cycle is induced, the same follows for any subdivision thereof. 
\end{proof}

We need to prove the following lemmas for non-planar triconnected components. We state the lemmas here and prove them in subsequent sections. 

\begin{lemma}\label{lemma:UICinK5} 
Every strongly-connected bipartite subdivision of $K_5$ has an undominated induced cycle. 
\end{lemma}

\begin{lemma}\label{lemma:UICinV8} 
Every strongly-connected bipartite subdivision of $V_8$ has an undominated induced cycle. 
\end{lemma}

Finally, we need to stitch together various cycles across triconnected components. 

\begin{lemma}\label{lemma:cycleStitch}
Given the triconnected-component-separating-pair tree $T$ of (the underlying undirected graph of) a strongly connected bipartite graph $G$, it is possible to find an undominated induced cycle in the original graph $G$ in \Log. 
\end{lemma}

Thus, we obtain our main result: 

\begin{theorem}\label{theorem:main}
There is an \NLog\ (and therefore polynomial time) procedure for finding a Nash equilibrium in a two-player win-lose game which belongs to one of the following classes:
\begin{enumerate}
\item $K_{3,3}$ minor-free games, 
\item $K_5$-minor free games where the triconnected components are either planar or $V_8$, 
\item Games whose triconnected components are $K_5$, $V_8$ or planar. 
\end{enumerate}
For classes (1) and (2), a Nash equilibrium can, in fact, be computed in \ULog.
\end{theorem}

\begin{proof} We are done by invoking the above Lemmata. Apart from Lemmata~\ref{lemma:decomp3Conn}, \ref{lemma:UICin3connPlan}, \ref{lemma:cycleStitch} and \ref{lemma:UICNash}, we invoke \ref{lemma:UICinK5} and  \ref{lemma:K33_tricon} for the proof for class (1), \ref{lemma:UICinV8} and \ref{lemma:K5_tricon} for (2) and for (3), \ref{lemma:UICinK5}, \ref{lemma:UICinV8} and the observation that undominated induced cycles in $K_5$'s and $V_8$'s can be ``stitched'' together in \Log\ as well.  
\end{proof}

\begin{remark} 
Class (3) in Theorem~\ref{theorem:main} above includes games that are neither $K_{3,3}$-minor-free nor $K_5$-minor-free, and is a superclass of classes (1) and (2) (and of planar games) as well. 
\end{remark}


\subsection{Stitching cycles together}

Before we complete the proof of Lemma~\ref{lemma:cycleStitch}, we prove the following lemma:

\begin{lemma}\label{lemma:twoCycleStitch} Given a graph $G$ with two triconnected components $C_1$ and $C_2$ which share a separating pair $S = \{u,v\}$ and suppose $O_i \in C_i$ (for $i \in \{1,2\}$) are two undominated induced cycles both passing through the (undirected) edge $\{u,v\}$ but in opposite directions then the cycles can be ``stitched'' together to obtain an undominated induced cycle in $G$. 
\end{lemma}

\begin{proof}
It is easy to see that dropping the oriented copies of $S$, we get a directed cycle $O$ consisting of all other directed edges in $O_1$ and $O_2$. $O$ is an undominated induced cycle as there are no edges between vertices in $C_1$ and $C_2$, and $O_i$ is already undominated and induced in $C_i$ (for $i \in \{1,2\}$). 
\end{proof}

Now, we prove Lemma~\ref{lemma:cycleStitch}. 
\begin{proof}(of Lemma~\ref{lemma:cycleStitch})
We first modify the triconnected-component-separating-pair tree $T$ as obtained from Lemma~\ref{lemma:decomp3Conn}, as follows. 

Consider a separating pair $S = \{u, v\}$ and the corresponding triconnected components, say, $C_1, C_2, \ldots, C_i$. If component $C_j$ ($1 \leq j \leq i$) has a vertex $w$ that dominates $S$ (that is, $G$ has edges from $w$ to both $u$ and $v$), then we remove edges (in $T$) between $S$ and $C_k$, $\forall k \neq j$. If more than one component contains such vertices that dominate $S$, we pick one component arbitrarily. If none of them contain vertices that dominate $S$, then we leave the edges between $S$ and $C_k$ ($1 \leq k \leq i$) untouched. Repeat the above for all separating pairs. 

Now, we pick a tree, say $T_1$ from the resulting forest. It suffices to find an undominated induced cycle in $T_1$ because of the following. If $C_j$ ($1 \leq j \leq i$) has a vertex $w$ that dominates $S$, then undominated cycle(s) in $C_j$ do not pass through $u$ and $v$. Hence, no vertex outside $T_1$ can dominate any cycle within $T_1$. 

Notice that we can stitch together only two cycles at a separating pair. Therefore, we further modify $T_1$ as follows. We arbitrarily pick one undominated induced cycle per component in $T_1$. We keep the edge $(C,S)$ in $T_1$ if and only if the picked cycle in triconnected component $C$ passes through the separating pair $S$. If an $S$ node has just one  edge incident on it, then remove it. If it has at least two bidirected edges incident on it, keep exactly two of them and remove the other edges. Note that vertices from components we have ``disconnected'' from $S$ cannot dominate (in $G$) cycles in the components we have retained. This is because each ``disconnected'' component itself has an undominated cycle passing though $S$. 

The trees in the resulting forest are just paths. Pick one of these paths, say $P$. If $P$ has just one component $C$, we are done as any undominated induced cycle in $C$ is an undominated induced cycle in $G$ as well. On the other hand, if $P$ has two or more components, we repeatedly invoke Lemma~\ref{lemma:twoCycleStitch} to stitch together undominated induced cycles across all these components, thereby obtaining an undominated induced cycle in $G$.  

Since the steps of the above procedure can be performed by a \Log-transducer, we have completed the proof of the lemma. 
\end{proof}


\subsection{Undominated Induced Cycle in $K_5$}

We prove Lemma~\ref{lemma:UICinK5}, that is, every strongly connected bipartite (oriented) subdivision of $K_5$ has an undominated induced cycle. The key idea is that removing any one edge of a $K_5$ makes it planar. 

\begin{proof} (of Lemma~\ref{lemma:UICinK5}) 
Let $G$ be a strongly connected bipartite (oriented) subdivision of $K_5$. Using Lemma~\ref{lemma:UICinK5minuse} below, there exists an edge (or a subdivided edge) $e = (w_1, w_2)$ such that $G \setminus e$ is strongly connected. As $G \setminus e$ is also planar, by the result of Addario-Berry, Olver and Vetta \cite{AOV07} (Lemma~\ref{lemma:mainInAOV} above), it has an undominated induced cycle, say $U$. Notice that $U$ continues to be induced in $G$. We also prove in Lemma~\ref{lemma:UICinK5minuse} below that adding back $e = (w_1, w_2)$ introduces no new domination(s) in $G$, because either $w_1$ and $w_2$ belong to the same color class ($R$ or $C$) or $w_2$ is the only vertex in its color class. It follows that $U$ is undominated in $G$ as well. 
\end{proof}


\begin{lemma}\label{lemma:UICinK5minuse} 
Let $G = (V, E)$ be an orientation of $K_5$ that is strongly connected and such that each $w \in V$ belongs to one of the color classes $R$ or $C$. There exists an edge $e \in E$, $e = (w_1, w_2)$ such that $G' = G \setminus e$ is strongly connected and such that one of the following holds: 
\begin{enumerate}
\item $w_1$ and $w_2$ belong to the same color class. 
\item One of the color classes (say, $C$) has exactly one vertex and $e$ is an incoming edge to that vertex. That is, $C = \{w_2\}$. 
\end{enumerate}
\end{lemma}
\begin{proof} Let $V = \{1, 2, 3, 4, 5\}$. Consider the color classes ($R$ or $C$) to which these vertices $1, 2, 3, 4, 5$ belong. It suffices to look at the number of vertices belonging to each class: 
\begin{enumerate}
\item[(I)] All the vertices belong to the same color class. Without loss of generality, let $1, 2, 3, 4, 5 \in R$. In this case, it suffices to prove that $\exists e = (w_1, w_2)\in E$ such that $G' = G \setminus e$ is strongly connected. Condition (i) of the theorem holds as $w_1$ and $w_2$ belong to the same color class, no matter what $e$ we choose. 
\item[(II)] $|R| = 4, |C| = 1$. Here, apart from proving that $G \setminus e$ is strongly connected, we need to prove that either $w_1, w_2 \in R$ or $C = \{w_2\}$. 
\item[(III)] $|R| = 3, |C| = 2$. We need to prove that $G \setminus e$ is strongly connected and $w_1, w_2 \in R$ or $w_1, w_2 \in C$.
\end{enumerate}

If there does not exist any edge whose removal ensures that the rest of the graph is strongly connected, then $\forall e = (u, v) \in E$, either $u$ is a source in $G \setminus e$ or $v$ is a sink. Without loss of generality, let $v$ be a sink. (The case when $u$ is a source is symmetric and a similar argument works). Consider the subgraph $G \setminus u$. As shown in Figure 1, only two non-isomorphic cases arise. (We denote the remaining vertices by $a, b, c$). 
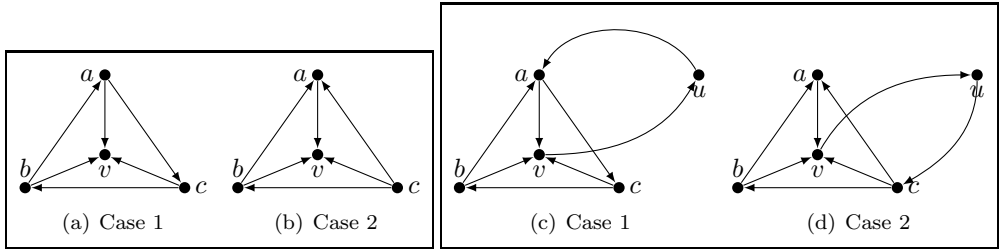
\begin{figure}
\begin{frame}{
\subfigure[Case 1]{
\begin{tikzpicture}[inner sep =1.5pt]
\node (A) at (0,0) [circle, fill, label=left:$a$]{};
\node (B) at (-1.0605,-1.5) [circle, fill,label=above:$b$]{};
\node (C) at (1.0605,-1.5) [circle, fill,label=right:$c$]{};
\node (V) at (0,-1.0605) [circle, fill,label=below:$v$]{};
\path[>=latex, ->] (B) edge (A); 
\path[>=latex, ->] (C) edge (B); 
\path[>=latex, ->] (A) edge (C); 
\path[>=latex, ->] (A) edge (V); 
\path[>=latex, ->] (B) edge (V); 
\path[>=latex, ->] (C) edge (V); 
\end{tikzpicture}
}
\subfigure[Case 2]{
\begin{tikzpicture}[inner sep =1.5pt]
\node (A) at (0,0) [circle, fill, label=left:$a$]{};
\node (B) at (-1.0605,-1.5) [circle, fill,label=above:$b$]{};
\node (C) at (1.0605,-1.5) [circle, fill,label=right:$c$]{};
\node (V) at (0,-1.0605) [circle, fill,label=below:$v$]{};
\path[>=latex, ->] (B) edge (A); 
\path[>=latex, ->] (C) edge (B); 
\path[>=latex, ->] (C) edge (A); 
\path[>=latex, ->] (A) edge (V); 
\path[>=latex, ->] (B) edge (V); 
\path[>=latex, ->] (C) edge (V); 
\end{tikzpicture}
}
}\end{frame}
\begin{frame}{
\subfigure[Case 1]{
\begin{tikzpicture}[inner sep =1.5pt]
\node (A) at (0,0) [circle, fill, label=left:$a$]{};
\node (B) at (-1.0605,-1.5) [circle, fill,label=above:$b$]{};
\node (C) at (1.0605,-1.5) [circle, fill,label=right:$c$]{};
\node (U) at (2.121,0) [circle, fill,label=below:$u$]{};
\node (V) at (0,-1.0605) [circle, fill,label=below:$v$]{};
\path[>=latex, ->] (B) edge (A); 
\path[>=latex, ->] (C) edge (B); 
\path[>=latex, ->] (A) edge (C); 
\path[>=latex, ->] (A) edge (V); 
\path[>=latex, ->] (B) edge (V); 
\path[>=latex, ->] (C) edge (V); 
\draw[>=latex, ->] (V) to [out=0,in=-120] (U); 
\draw[>=latex, ->] (U) to [out=120,in=60] (A); 
\end{tikzpicture}
}
\subfigure[Case 2]{
\begin{tikzpicture}[inner sep =1.5pt]
\node (A) at (0,0) [circle, fill, label=left:$a$]{};
\node (B) at (-1.0605,-1.5) [circle, fill,label=above:$b$]{};
\node (C) at (1.0605,-1.5) [circle, fill,label=right:$c$]{};
\node (U) at (2.121,0) [circle, fill,label=below:$u$]{};
\node (V) at (0,-1.0605) [circle, fill,label=below:$v$]{};
\path[>=latex, ->] (B) edge (A); 
\path[>=latex, ->] (C) edge (B); 
\path[>=latex, ->] (C) edge (A); 
\path[>=latex, ->] (A) edge (V); 
\path[>=latex, ->] (B) edge (V); 
\path[>=latex, ->] (C) edge (V); 
\draw[>=latex, ->] (V) to [out=60,in=180] (U); 
\draw[>=latex, ->] (U) to [out=-90,in=30] (C); 
\end{tikzpicture}
}
}\end{frame}
\caption{(a) and (b) show cases 1 and 2 in the proof for $K_5$-subdivision. These are the only two non-isomorphic orientations of $G \setminus u$ where $G$ is a strongly connected orientation of $K_5$ with vertices \{a,b,c,u,v\}. (c) and (d): After adding the vertex $u$ and some edges to Case 1 and Case 2 respectively.}
\end{figure}

Consider case 1. \\As G is strongly connected, $(v, u) \in E$ and at least one of $(u, a), (u, b), (u, c) \in E$. 
Without loss of generality, let $(u, a) \in E$. This is shown in Figure 1(c). 
Now, in case (I), independent of the orientation of the edges between $u$, $b$ and $u$, $c$, the graphs $G \setminus (a,v)$, $G \setminus (b,v)$ and $G \setminus (c,v)$ are strongly connected. 
In case (II), if $C = \{v\}$ or $C = \{u\}$, we may again choose $e$ to be either $(a,v)$, $(b,v)$ or $(c,v)$. If $C = \{a\}$, then we may drop $(b,v)$ or $(c,v)$ so that condition (i) is satisfied. Similarly for $C = \{b\}$ or $C = \{c\}$. 
In case (III), at least one of $a, b, c$ is in $R$. Say $a \in R$. If $v \in R$, then we can choose $e = (a, v)$. If $v \in C$ and one of $a, b, c \in C$ (say $b \in C$), we choose $e = (b, v)$. If $C = \{u, v\}$ then choose $e = (b, a)$. 

Now, consider case 2. Again as G is strongly connected, $(v, u) \in E$ and $(u, c) \in E$. Refer Figure 1(d).
It is easy to see that if we choose $e = (c, a)$ or $(b, a)$ or $(c, v)$ or $(b, v)$, $G \setminus e$ is strongly connected. The proof for case (I) follows rightaway. In cases (II) and (III) too, we are done as one of these pairs of vertices must belong to the same color class. 
\end{proof}


\subsection{Undominated Induced Cycle in $V_8$}\label{subsection:V8proof}

\begin{proof}(of Lemma~\ref{lemma:UICinV8}) 
Let $G = (V, E)$ be a counter example and let $G' = (V', E')$ be its underlying $V_8$. ($G'$ is obtained from $G$ by ignoring the subdividing vertices). \\
As $G$ is strongly connected, so is $G'$. On the other hand, though $G$ is bipartite, $G'$ is not. (It is easy to see that $G$ must have at least 3 sub-dividing vertices). Also, $G'$ has no undominated induced cycle (otherwise $G$ has one too).  
We claim the following, which are easy to check. (Refer Appendix~\ref{appendix:claim4} for an outline of the proof). \\

{\noindent}{\textbf{Claim 4. }} The following hold for cycles in $G'$ (and in $G$): 
\begin{enumerate}
\item All cycles in $G'$ (and hence in $G$) are of length $> 3$. 
\item $4$-cycles in $G'$ correspond to undominated induced cycles in $G$. 
\item $5$-cycles in $G'$ are induced but may be dominated (in both $G'$ and $G$). 
\item If $G'$ has a $6$-cycle, it has a $4$-cycle too. 
\item If $G'$ has a $7$-cycle, it either has a $4$-cycle or a $5$-cycle. 
\item If $G'$ has an $8$-cycle, it has a $5$-cycle too. \\
\end{enumerate}

\noindent As $G$ is a counter-example, because of Claim 4 above, $G'$ contains no $4$-cycle. 

$\implies$ $G'$ contains a $5$-cycle. \\
Without loss of generality, let this cycle be $(a, b, c, d, e, a)$. As $G$ is a counter example, this $5$-cycle is dominated (in both $G$ and $G'$). This cycle may either be dominated by $f$ or by $h$. This leads to two cases as shown in Figure 2. \\

\begin{figure}
\begin{frame}
{
\subfigure[Case 1]{
\begin{tikzpicture}[inner sep =1.5pt]
\node (A) at (0,0) [circle, fill, label=above:$a$]{};
\node (B) at (1,0) [circle, fill,label=above:$b$]{};
\node (C) at (1.707,-0.707) [circle, fill,label=right:$c$]{};
\node (D) at (1.707,-1.707) [circle, fill,label=right:$d$]{};
\node (E) at (1,-2.414) [circle, fill,label=below:$e$]{};
\node (F) at (0,-2.414) [circle, fill,label=below:$f$]{};
\node (G) at (-0.707,-1.707) [circle, fill,label=left:$g$]{};
\node (H) at (-0.707,-0.707) [circle, fill,label=left:$h$]{};
\path[>=latex, ->] (A) edge (B); 
\path[>=latex, ->] (B) edge (C); 
\path[>=latex, ->] (C) edge (D); 
\path[>=latex, ->] (D) edge (E); 
\path[>=latex, -] (F) edge (G); 
\path[>=latex, -] (G) edge (H); 
\path[>=latex, -] (H) edge (A); 
\path[>=latex, ->] (E) edge (A); 
\path[>=latex, -] (H) edge (D); 
\path[>=latex, -] (C) edge (G); 
\path[>=latex, ->] (F) edge (B); 
\path[>=latex, ->] (F) edge (E); 
\end{tikzpicture}
}
\subfigure[Case 2]{
\begin{tikzpicture}[inner sep =1.5pt]
\node (A) at (0,0) [circle, fill, label=above:$a$]{};
\node (B) at (1,0) [circle, fill,label=above:$b$]{};
\node (C) at (1.707,-0.707) [circle, fill,label=right:$c$]{};
\node (D) at (1.707,-1.707) [circle, fill,label=right:$d$]{};
\node (E) at (1,-2.414) [circle, fill,label=below:$e$]{};
\node (F) at (0,-2.414) [circle, fill,label=below:$f$]{};
\node (G) at (-0.707,-1.707) [circle, fill,label=left:$g$]{};
\node (H) at (-0.707,-0.707) [circle, fill,label=left:$h$]{};
\path[>=latex, ->] (A) edge (B); 
\path[>=latex, ->] (B) edge (C); 
\path[>=latex, ->] (C) edge (D); 
\path[>=latex, ->] (D) edge (E); 
\path[>=latex, -] (F) edge (G); 
\path[>=latex, -] (G) edge (H); 
\path[>=latex, ->] (E) edge (A); 
\path[>=latex, -] (C) edge (G); 
\path[>=latex, -] (F) edge (B); 
\path[>=latex, -] (F) edge (E); 
\path[>=latex, ->] (H) edge (A); 
\path[>=latex, ->] (H) edge (D); 
\end{tikzpicture}
}
}\end{frame}
\begin{frame}{
\subfigure[Case 1]{
\begin{tikzpicture}[inner sep =1.5pt]
\node (A) at (0,0) [circle, fill, label=above:$a$]{};
\node (B) at (1,0) [circle, fill,label=above:$b$]{};
\node (C) at (1.707,-0.707) [circle, fill,label=right:$c$]{};
\node (D) at (1.707,-1.707) [circle, fill,label=right:$d$]{};
\node (E) at (1,-2.414) [circle, fill,label=below:$e$]{};
\node (F) at (0,-2.414) [circle, fill,label=below:$f$]{};
\node (G) at (-0.707,-1.707) [circle, fill,label=left:$g$]{};
\node (H) at (-0.707,-0.707) [circle, fill,label=left:$h$]{};
\path[>=latex, ->] (A) edge (B); 
\path[>=latex, ->] (B) edge (C); 
\path[>=latex, ->] (C) edge (D); 
\path[>=latex, ->] (D) edge (E); 
\path[>=latex, <-] (F) edge (G); 
\path[>=latex, <-] (G) edge (H); 
\path[>=latex, <-] (H) edge (A); 
\path[>=latex, ->] (E) edge (A); 
\path[>=latex, ->] (H) edge (D); 
\path[>=latex, <-] (C) edge (G); 
\path[>=latex, ->] (F) edge (B); 
\path[>=latex, ->] (F) edge (E); 
\end{tikzpicture}
}
\subfigure[Case 2]{
\begin{tikzpicture}[inner sep =1.5pt]
\node (A) at (0,0) [circle, fill, label=above:$a$]{};
\node (B) at (1,0) [circle, fill,label=above:$b$]{};
\node (C) at (1.707,-0.707) [circle, fill,label=right:$c$]{};
\node (D) at (1.707,-1.707) [circle, fill,label=right:$d$]{};
\node (E) at (1,-2.414) [circle, fill,label=below:$e$]{};
\node (F) at (0,-2.414) [circle, fill,label=below:$f$]{};
\node (G) at (-0.707,-1.707) [circle, fill,label=left:$g$]{};
\node (H) at (-0.707,-0.707) [circle, fill,label=left:$h$]{};
\path[>=latex, ->] (A) edge (B); 
\path[>=latex, ->] (B) edge (C); 
\path[>=latex, ->] (C) edge (D); 
\path[>=latex, ->] (D) edge (E); 
\path[>=latex, ->] (F) edge (G); 
\path[>=latex, ->] (G) edge (H); 
\path[>=latex, ->] (E) edge (A); 
\path[>=latex, -] (C) edge (G); 
\path[>=latex, <-] (F) edge (B); 
\path[>=latex, <-] (F) edge (E); 
\path[>=latex, ->] (H) edge (A); 
\path[>=latex, ->] (H) edge (D); 
\end{tikzpicture}
}
}\end{frame}
\caption{Orienting edges in $G'$ (in the proof for $V_8$-subdivision), given that the cycle $(a, b, c, d, e, a)$ is dominated. (a) Case 1: Cycle dominated by $f$. (b) Case 2: Cycle dominated by $h$. (c) Case 1 forces this orientation and this has an undominated induced cycle $(a, h, d, e, a)$. (d) Case 2 forces this orientation. No matter how the edge between $c$ and $g$ is oriented, there is an undominated induced cycle.}
\end{figure}
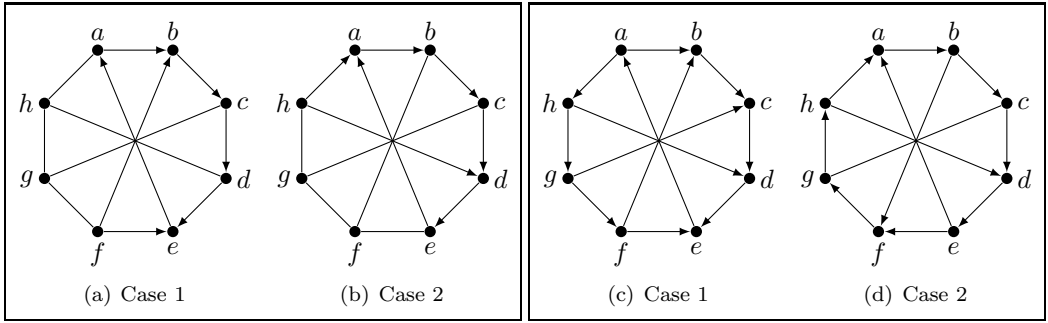
\noindent Case 1: $f$ dominates $(a, b, c, d, e, a)$. That is, $(f, b)$ and $(f, e) \in E'$. 

$\implies$ $(g, f) \in E'$ (otherwise $f$ is a source in $G'$ but $G'$ is strongly connected). 

$\implies$ $(g, c) \in E'$ (otherwise $(c, g, f, b, c)$ is a $4$-cycle, hence undominated, induced). 

$\implies$ $(h, g) \in E'$ (owing to strong-connectivity, again)

$\implies$ $(h, d) \in E'$ (otherwise we have a $4$-cycle, again)

$\implies$ $(a, h) \in E'$ (due to strong connectivity)

$\implies$ $(a, h, d, e, a)$ is a $4$-cycle (an undominated induced cycle). \\

\noindent This contradicts the fact that $G$ is a counter example. Therefore, $f$ cannot dominate $(a, b, c, d, e, a)$. \\

\noindent Case 2: $h$ dominates $(a, b, c, d, e, a)$. That is, $(h, a)$ and $(h, d) \in E'$. 

$\implies$ $(g, h) \in E'$ (since $G'$ is strongly connected). \\

\noindent Now, $(f, g) \in E'$ (because, if $(c, g) \in E'$, $f$ has to dominate the resulting $5$-cycle, and if $(g, c) \in E'$, $(f, g)$ ensures strong-connectivity). Similarly, $(e, f) \in E'$. 

$\implies$ $(b, f) \in E'$ (otherwise $(f, b, c, d, e, f)$ is an undominated induced cycle). \\

\noindent Now, if $(g, c) \in E'$, $(h, d, e, f, g, h)$ is an undominated induced cycle. Otherwise, $(c, g, h, a, b, c)$ is an undominated induced cycle. \\

\noindent Hence, $h$ cannot dominate $(a, b, c, d, e, a)$ either.

$\implies$ $(a, b, c, d, e, a)$ itself is an undominated induced cycle. 

\end{proof}


\section{Conclusion and Future Work}\label{section:futurework}

We have sharpened the result of Addario-Berry, Olver and Vetta \cite{AOV07} by observing that their polynomial time algorithm for planar win-lose bimatrix games is, in fact, in \Log. Further, we have extended their result and shown that the following classes of win-lose bimatrix games (which include planar games as well, and are classes that include non-trivial non-planar games too) are in \NLog\ (and hence in \Ptime). In fact, classes (1) and (2) are in \ULog. 
\begin{enumerate}
\item $K_{3,3}$ minor-free games, 
\item $K_5$-minor free games where the triconnected components are either planar or $V_8$, 
\item Games whose triconnected components are $K_5$, $V_8$ or planar. 
\end{enumerate}

As win-lose bimatrix games are PPAD hard and as we have proved polynomial time tractability of the above classes, finding or narrowing down on what actually causes the hardness is an open problem. It will be interesting to prove a necessary and sufficient condition for hardness of win-lose games. 

\section*{Acknowledgement}
The authors wish to thank T. Parthasarathy, K. V. Subrahmanyam, G. Ravindran, Prajakta Nimbhorkar and Tanmoy Chakraborty for useful discussions. The first author would also like to thank the organizers and participants of the Aarhus-Chennai Computational Complexity Workshop 2010, for their feedback on a talk he had given, on the topic. 




\newpage\appendix

\section{Proof of Claim 4.}\label{appendix:claim4}

\begin{figure}
\subfigure[]{
\begin{tikzpicture}[inner sep =1.5pt]
\node (A) at (0,0) [circle, fill, label=above:$i$]{};
\node (B) at (1,0) [circle, fill,label=above:$i+1$]{};
\node (C) at (1.707,-0.707) [circle, fill,label=right:$i+2$]{};
\node (D) at (1.707,-1.707) [circle, fill,label=right:$i+3$]{};
\node (E) at (1,-2.414) [circle, fill,label=below:$i+4$]{};
\node (F) at (0,-2.414) [circle, fill,label=below:$i+5$]{};
\node (G) at (-0.707,-1.707) [circle, fill,label=left:$i+6$]{};
\node (H) at (-0.707,-0.707) [circle, fill,label=left:$i+7$]{};
\path[>=latex, ->] (A) edge (B); 
\path[>=latex, -] (B) edge (C); 
\path[>=latex, -] (C) edge (D); 
\path[>=latex, -] (D) edge (E); 
\path[>=latex, -] (F) edge (G); 
\path[>=latex, -] (G) edge (H); 
\path[>=latex, -] (H) edge (A); 
\path[>=latex, ->] (E) edge (A); 
\path[>=latex, -] (H) edge (D); 
\path[>=latex, -] (C) edge (G); 
\path[>=latex, <-] (F) edge (B); 
\path[>=latex, ->] (F) edge (E); 
\end{tikzpicture}
}
\subfigure[]{
\begin{tikzpicture}[inner sep =1.5pt]
\node (A) at (0,0) [circle, fill, label=above:$i$]{};
\node (B) at (1,0) [circle, fill,label=above:$i+1$]{};
\node (C) at (1.707,-0.707) [circle, fill,label=right:$i+2$]{};
\node (D) at (1.707,-1.707) [circle, fill,label=right:$i+3$]{};
\node (E) at (1,-2.414) [circle, fill,label=below:$i+4$]{};
\node (F) at (0,-2.414) [circle, fill,label=below:$i+5$]{};
\node (G) at (-0.707,-1.707) [circle, fill,label=left:$i+6$]{};
\node (H) at (-0.707,-0.707) [circle, fill,label=left:$i+7$]{};
\path[>=latex, ->] (A) edge (B); 
\path[>=latex, ->] (B) edge (C); 
\path[>=latex, ->] (C) edge (D); 
\path[>=latex, ->] (D) edge (E); 
\path[>=latex, -] (F) edge (G); 
\path[>=latex, -] (G) edge (H); 
\path[>=latex, -] (H) edge (A); 
\path[>=latex, ->] (E) edge (A); 
\path[>=latex, -] (H) edge (D); 
\path[>=latex, -] (C) edge (G); 
\path[>=latex, -] (F) edge (B); 
\path[>=latex, -] (F) edge (E); 
\end{tikzpicture}
}
\subfigure[]{
\begin{tikzpicture}[inner sep =1.5pt]
\node (A) at (0,0) [circle, fill, label=above:$i$]{};
\node (B) at (1,0) [circle, fill,label=above:$i+1$]{};
\node (C) at (1.707,-0.707) [circle, fill,label=right:$i+2$]{};
\node (D) at (1.707,-1.707) [circle, fill,label=right:$i+3$]{};
\node (E) at (1,-2.414) [circle, fill,label=below:$i+4$]{};
\node (F) at (0,-2.414) [circle, fill,label=below:$i+5$]{};
\node (G) at (-0.707,-1.707) [circle, fill,label=left:$i+6$]{};
\node (H) at (-0.707,-0.707) [circle, fill,label=left:$i+7$]{};
\path[>=latex, ->] (A) edge (B); 
\path[>=latex, ->] (B) edge (C); 
\path[>=latex, -] (C) edge (D); 
\path[>=latex, -] (D) edge (E); 
\path[>=latex, <-] (F) edge (G); 
\path[>=latex, -] (G) edge (H); 
\path[>=latex, -] (H) edge (A); 
\path[>=latex, ->] (E) edge (A); 
\path[>=latex, -] (H) edge (D); 
\path[>=latex, ->] (C) edge (G); 
\path[>=latex, -] (F) edge (B); 
\path[>=latex, ->] (F) edge (E); 
\end{tikzpicture}
}

\subfigure[]{
\begin{tikzpicture}[inner sep =1.5pt]
\node (A) at (0,0) [circle, fill, label=above:$i$]{};
\node (B) at (1,0) [circle, fill,label=above:$i+1$]{};
\node (C) at (1.707,-0.707) [circle, fill,label=right:$i+2$]{};
\node (D) at (1.707,-1.707) [circle, fill,label=right:$i+3$]{};
\node (E) at (1,-2.414) [circle, fill,label=below:$i+4$]{};
\node (F) at (0,-2.414) [circle, fill,label=below:$i+5$]{};
\node (G) at (-0.707,-1.707) [circle, fill,label=left:$i+6$]{};
\node (H) at (-0.707,-0.707) [circle, fill,label=left:$i+7$]{};
\path[>=latex, ->] (A) edge (B); 
\path[>=latex, ->] (B) edge (C); 
\path[>=latex, -] (C) edge (D); 
\path[>=latex, ->] (D) edge (E); 
\path[>=latex, -] (F) edge (G); 
\path[>=latex, ->] (G) edge (H); 
\path[>=latex, -] (H) edge (A); 
\path[>=latex, ->] (E) edge (A); 
\path[>=latex, ->] (H) edge (D); 
\path[>=latex, ->] (C) edge (G); 
\path[>=latex, -] (F) edge (B); 
\path[>=latex, -] (F) edge (E); 
\end{tikzpicture}
}
\subfigure[]{
\begin{tikzpicture}[inner sep =1.5pt]
\node (A) at (0,0) [circle, fill, label=above:$i$]{};
\node (B) at (1,0) [circle, fill,label=above:$i+1$]{};
\node (C) at (1.707,-0.707) [circle, fill,label=right:$i+2$]{};
\node (D) at (1.707,-1.707) [circle, fill,label=right:$i+3$]{};
\node (E) at (1,-2.414) [circle, fill,label=below:$i+4$]{};
\node (F) at (0,-2.414) [circle, fill,label=below:$i+5$]{};
\node (G) at (-0.707,-1.707) [circle, fill,label=left:$i+6$]{};
\node (H) at (-0.707,-0.707) [circle, fill,label=left:$i+7$]{};
\path[>=latex, ->] (A) edge (B); 
\path[>=latex, ->] (B) edge (C); 
\path[>=latex, ->] (C) edge (D); 
\path[>=latex, ->] (D) edge (E); 
\path[>=latex, ->] (F) edge (G); 
\path[>=latex, ->] (G) edge (H); 
\path[>=latex, ->] (H) edge (A); 
\path[>=latex, -] (E) edge (A); 
\path[>=latex, -] (H) edge (D); 
\path[>=latex, -] (C) edge (G); 
\path[>=latex, -] (F) edge (B); 
\path[>=latex, <-] (F) edge (E); 
\end{tikzpicture}
}
\subfigure[]{
\begin{tikzpicture}[inner sep =1.5pt]
\node (A) at (0,0) [circle, fill, label=above:$i$]{};
\node (B) at (1,0) [circle, fill,label=above:$i+1$]{};
\node (C) at (1.707,-0.707) [circle, fill,label=right:$i+2$]{};
\node (D) at (1.707,-1.707) [circle, fill,label=right:$i+3$]{};
\node (H) at (1,-2.414) [circle, fill,label=below:$i+7$]{};
\node (G) at (0,-2.414) [circle, fill,label=below:$i+6$]{};
\node (F) at (-0.707,-1.707) [circle, fill,label=left:$i+5$]{};
\node (E) at (-0.707,-0.707) [circle, fill,label=left:$i+4$]{};
\path[>=latex, ->] (A) edge (B); 
\path[>=latex, ->] (B) edge (C); 
\path[>=latex, ->] (C) edge (D); 
\path[>=latex, -] (D) edge (E); 
\path[>=latex, <-] (F) edge (G); 
\path[>=latex, <-] (G) edge (H); 
\path[>=latex, -] (H) edge (A); 
\path[>=latex, ->] (E) edge (A); 
\path[>=latex, <-] (H) edge (D); 
\path[>=latex, -] (C) edge (G); 
\path[>=latex, -] (F) edge (B); 
\path[>=latex, ->] (F) edge (E); 
\end{tikzpicture}
}

\caption{(a) 4-cycles look like this. (b) 5-cycles look like this. (c) 6-cycles look like this. (d) 7-cycles look like this. (e) One ``type'' of 8-cycles. (f) The other non-isomorphic ``type'' of 8-cycles.}
\end{figure}
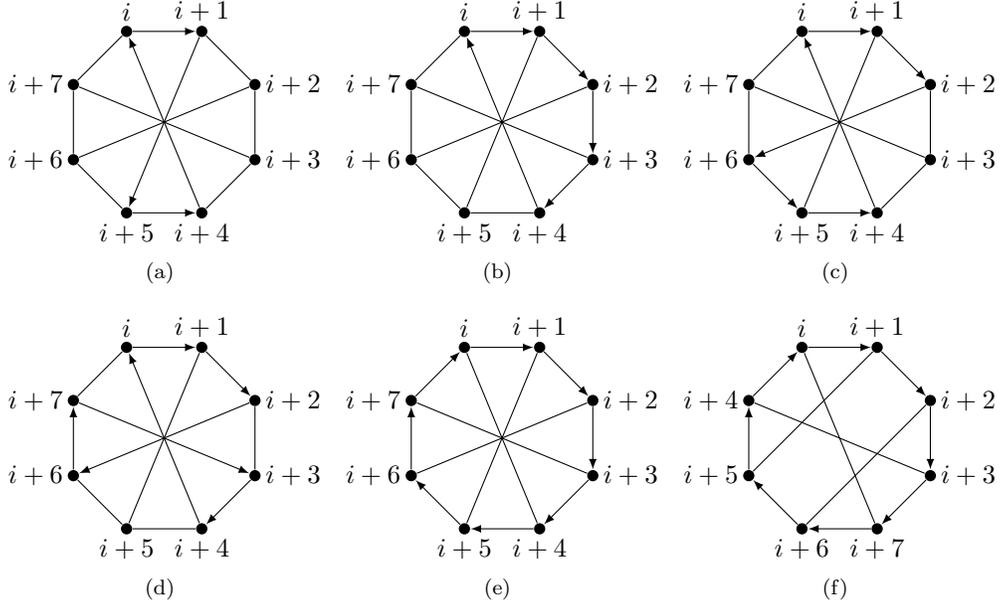


We discuss an informal outline of the proof of Claim 4 (in Subsection~\ref{subsection:V8proof}). 

\begin{proof}
It is easy to see that: \\

{\noindent} Subclaim (I): Given any vertex $j$ of $G'$, vertices adjacent to $j$ are $j-1, j+1$ and $j+4$ (all operations are modulo $8$). \\

{\noindent} (1) The claim that all cycles in $G'$ (and hence in $G$) are of length $> 3$, is easy to see. \\

{\noindent} (2) $4$-cycles in $G'$ are of the form $(i, i+1, i+5, i+4, i)$ (mod $8$, again). Refer Figure 3(a). (We are leaving out the symmetric cases $(i, i-1, i+3, i+4, i)$). 
Such cycles are, clearly, induced. Moreover, all vertices not on this cycle have only one edge incident on this cycle (because of Subclaim (I)). \\

{\noindent} (3) $5$-cycles in $G'$ are of the form $(i, i+1, i+2, i+3, i+4, i)$ (mod $8$). Refer Figure 3(b). Again, using Subclaim (I), there are no other edges between vertices on this cycle.  \\

{\noindent} (4) $6$-cycles in $G'$ are of the form $(i, i+1, i+2, i+6, i+5, i+4, i)$. Refer Figure 3(c). Here, if $(i+1, i+5) \in E'$, then $(i, i+1, i+5, i+4, i)$ is a 4-cycle. 
Otherwise, $(i+5, i+1) \in E'$ and $(i+1, i+2, i+6, i+5, i+1)$ is a 4-cycle. \\

{\noindent} (5) $7$-cycles in $G'$ are of the form $(i, i+1, i+2, i+6, i+7, i+3, i+4, i)$. (We have left out the symmetric and isomorphic cases). Refer Figure 3(d). 
Now, if $(i+7, i)$ is an edge, then $(i+7, i, i+1, i+2, i+6, i+7)$ is a $5$-cycle. Otherwise, $(i, i+7)$ is an edge and $(i, i+7, i+3, i+4, i)$ is a $4$-cycle. \\

{\noindent} (6) $8$-cycles in $G'$ are either of the form $(i, i+1, i+2, i+3, i+4, i+5, i+6, i+7, i)$ or $(i, i+1, i+2, i+3, i+7, i+6, i+5, i+4, i)$. (Again, we have left out the symmetric and isomorphic cases). Refer Figures 3(e) and 3(f). \\

In the first case, $i$ and $i+4$ (for example) are ``opposite'' vertices. If $(i, i+4)$ is an edge, then $(i, i+4, i+5, i+6, i+7, i)$ is a $5$-cycle. Instead, if $(i+4, i)$ is an edge, then $(i+4, i, i+1, i+2, i+3, i+4)$ is a $5$-cycle. \\

In the second case, $i$ and $i+7$ are ``opposite'' vertices which, again, lead to a 5-cycle depending on the direction of the edge $\{i, i+7\}$. 

\end{proof}


\end{document}